\documentclass[journal]{IEEEtran}
\usepackage{amsmath,amsfonts,amsthm}
\usepackage{algorithmic}
\usepackage{algorithm}
\usepackage{array}
\usepackage{textcomp}
\usepackage{stfloats}
\usepackage{url}
\usepackage{verbatim}
\usepackage{graphicx}
\usepackage{cite}
\usepackage{color}
\usepackage{tikz}
\usepackage{circuitikz}
\usepackage{accents}

\newtheorem{lem}{Lemma}

\newtheorem{res}{Result}

\expandafter\def\expandafter\normalsize\expandafter{%
    \normalsize%
    \setlength\abovedisplayskip{4pt}%
    \setlength\belowdisplayskip{4pt}%
    \setlength\abovedisplayshortskip{-8pt}%
    \setlength\belowdisplayshortskip{2pt}%
}

\begin{document}

\title{Dimensional Scaling Laws for Continuous Fluid Antenna Systems}

\author{Peter J. Smith,~\IEEEmembership{Fellow,~IEEE,} Amy S. Inwood,~\IEEEmembership{Member,~IEEE,} \\ Michail Matthaiou,~\IEEEmembership{Fellow,~IEEE,} and Rajitha Senanayake,~\IEEEmembership{Member,~IEEE}

\thanks{This work was supported in part by the U.K. Engineering and Physical Sciences Research Council (EPSRC) (grant No. EP/X04047X/1). The work of M. Matthaiou has received funding from the European Research Council (ERC) under the European Union’s Horizon 2020 research and innovation programme (grant agreement No. 101001331).}

\thanks{P. J. Smith is with the School of Mathematics and Statistics, Victoria University of Wellington,  Wellington 6140, New Zealand, (e-mail: peter.smith@vuw.ac.nz).}

\thanks{A. S. Inwood and M. Matthaiou are with the Centre for Wireless Innovation (CWI), Queen’s University Belfast, Belfast BT3 9DT, U.K., (e-mail:  \{a.inwood, m.matthaiou\}@qub.ac.uk).}

\thanks{R. Senanayake is with the Department of Electrical and Electronic Engineering, University of Melbourne, Melbourne, VIC 3010, Australia, (email: rajitha.senanayake@unimelb.edu.au).}
}

\maketitle

\begin{abstract}

 Consider the signal-to-noise ratio (SNR) of a continuous fluid antenna system (CFAS) operating over a Rayleigh fading channel. In this paper, we extend traditional system assumptions and consider spatially coherent isotropic correlation, continuous positioning of the antenna rather than discrete, and the use of multi-dimensional space (1D, 2D and 3D). By focusing on the upper tail of the received SNR distribution (the high SNR probability (HSP)), we are able to derive asymptotically exact closed-form formulas for the HSP.  Finally, these results lead to scaling laws which describe the increase in the HSP as we employ more dimensions and the optimal CFAS dimensions.
\end{abstract}

\section{Introduction} \label{sec:intro}

In recent years, fluid antenna systems (FASs) have gained considerable research attention due to their ability to enhance the signal-to-noise ratio (SNR) of a system and mitigate interference \cite{kkwong1,kkwong2}. The advantages of FASs, demonstrated in \cite{kkwong1,kkwong2,davies} and elsewhere, have led to their proposal as a potential technology for sixth-generation (6G) wireless systems \cite{FAS_6G}.  A FAS has come to be known as an umbrella term for any antenna structure where flexible positioning is possible, whether it be implemented with fluid antennas (FAs) with dielectric holders, mechanically movable antennas or pixel-based switching. The fundamental advantage offered by all of these systems is the ability to make use of spatial variations in the channel across the potential locations of the antenna(s).

Most work in the area considers antennas which can be located at a finite number of positions \cite{kkwong1, FAS_6G,alouini}. However, the challenging and under-explored case of continuous space is also of interest and initial work in this area can be seen in \cite{psomas}. Also, most work considers antennas positioned along a line, but 2D and 3D FASs can lead to enhanced performance \cite{FAS_6G,kkwong3,kkwong4}. Hence, we consider positioning in a 1D line, a 2D rectangle and a 3D cuboid.

Motivated by these gaps, we consider a continuous FAS (CFAS) and employ a spatially consistent correlation model in 1D-3D. This is not always possible in discrete FAS systems due to analytical complexity. Instead, innovative approximations are used including block-correlation models \cite{ramirez-espinosa_new_2024} and copula models \cite{ghadi_copula_2023,ghadi_gaussian_2024}, as well as simpler correlation structures \cite{alouini}. In this work, we do not need to rely on any approximations for the correlation structure and  use the Jakes' model for all spatial correlations in 1D-3D.

Following \cite{psomas}, we focus on the peak performance offered by CFASs, i.e. the upper tail of the cumulative distribution function (CDF) of the received SNR. While analysis of the lower tail, giving outage probability, is highly desirable, severe outages require the SNR to be small over the whole range of possible antenna positions. Such events are referred to as small ball probabilities and are much harder to analyze than high SNR events \cite{li_gaussian_2001}. In contrast, analyzing the upper tail gives simple, closed-form results which are asymptotically exact and provide important performance insights.

In this paper, we make the following contributions:
\begin{itemize}
    \item We derive a highly accurate approximation to the high SNR probability (HSP), based on the received SNR for 1D, 2D and 3D CFASs.
    \item We propose an approximate dimensional scaling law that relates the HSP of a fixed antenna to that of an $n$-dimensional continuous FA (CFA).
    \item Finally, we provide the optimal shape that maximizes the HSP approximations for CFAs in rectangular 2D and cuboidal 3D spaces given limits on each dimension.
\end{itemize}

\textit{Notation}: Lower boldface letters represent vectors; $\mathbb{E}[\cdot]$ is the statistical expectation; $\mathrm{Var}[\cdot]$ is the variance; $P(A)$ is the probability of event $A$; $\mathcal{CN}(\mu,\sigma^2)$ is a complex Gaussian distribution with mean $\mu$ and variance $\sigma^2$; $\chi_k^2$ is a chi-squared distribution with $k$ degrees of freedom; $\mathrm{Exp}(\mu)$ is an exponential distribution with a mean of $\mu$; $J_0(\cdot)$ is the zeroth order Bessel function of the first kind; $\Gamma(\cdot)$ is the gamma function; $(\cdot)^*$ is the complex conjugate; $||\cdot||$ is the Euclidean norm and $\mathrm{dim}(A)$ is the number of dimensions of $A$.

\section{System Model}\label{sec:sysmodel}
Consider an antenna located at the coordinate $\mathbf{t}$ in a set of coordinate points, denoted $A$. As shown in Fig.~\ref{fig:dimensions}, we consider a range of scenarios for $A$, ranging from a fixed point to a 3D cuboid. Note that while the antenna position changes, its shape remains constant, so it is reasonable to assume that the propagation characteristics are constant. The channel from a single antenna source to the CFA is denoted $h(\mathbf{t})$, and in the case of Rayleigh fading we have $h(\mathbf{t})\sim \mathcal{CN}(0,\beta)$, where $\beta$ is the channel gain. We assume that the set $A$ is small enough so that $\beta$ is constant across $A$. Hence, the received signal at location $\mathbf{t}$ is given by
\begin{equation}
    r(\mathbf{t}) = h(\mathbf{t})s + n,
\end{equation}
where $s$ is the transmitted signal with $\mathbb{E}[|s|^2]=E_s$ and $n\sim\mathcal{CN}(0,\sigma^2)$ is the additive white Gaussian noise. Hence, the SNR at $\mathbf{t}$ is $\mathrm{SNR}(\mathbf{t})=\tfrac{E_s}{\sigma^2}|h(\mathbf{t})|^2$. The optimal SNR in a perfect CFAS is
\begin{equation}
    \mathrm{SNR} = \sup_{\mathbf{t}\in A}\left\{\frac{E_s}{\sigma^2}|h(\mathbf{t})|^2\right\}.
\end{equation}
Thus, we investigate the fundamental benefits of an antenna that can move in multi-dimensional space to maximize SNR.

For convenience, we normalize the channel so that $h(\mathbf{t}) = \sqrt{{\beta}/{2}}\,\Tilde{h}(\mathbf{t})$. Hence, $\mathbb{E}[|\Tilde{h}(\mathbf{t})|^2]=2$, and the real and imaginary components of $\Tilde{h}(\mathbf{t})$ both have unit power. This normalization makes it easier to use the random field theory in Section \ref{sec:analysis}. Using the normalized channel, we have
\begin{align}
    \mathrm{SNR} = \frac{\beta E_s}{2\sigma^2}\sup_{\mathbf{t}\in A} \left\{|\Tilde{h}(\mathbf{t})|^2\right\}
    = \frac{\beta E_s}{2\sigma^2}\sup_{\mathbf{t}\in A}\left\{X(\mathbf{t})\right\},
\end{align}
where $X(\mathbf{t})$ is a $\chi_2^2$ process. As discussed in Section \ref{sec:intro}, in this paper, we focus on the HSP defined by $P_{hs} = P(\mathrm{SNR}>u)$ for some high SNR threshold. Hence,
\begin{equation}
    P_{hs} = P\left(\sup_{\mathbf{t}\in A}\left\{X(\mathbf{t})\right\} \geq u_0\right), \label{eq:Ppp}
\end{equation}
where $u_0 = \frac{2\sigma^2u}{\beta E_s}$. The HSP in (\ref{eq:Ppp}) is a convenient formulation as it expresses the probability in terms of the supremum of a $\chi_2^2$ process over $A$ matching the random field notation in \cite{adler}.

In concrete terms we can think of four different physical layouts for (\ref{eq:Ppp}). These are defined in terms of their dimensionality below and are shown in Fig. \ref{fig:dimensions}.

\textbf{Fixed antenna case}: Here, $\mathbf{t}=0$ is a fixed point, and the antenna has no dimension in which to move.

\textbf{One-dimensional case:} Here, $\mathbf{t}=t\in A = [0,T_1]$. Hence, the antenna can move along a line of length $T_1$.

\textbf{Two-dimensional case:} Here, $\mathbf{t}=(t_1,t_2)\in A = [0,T_1] \times [0,T_2]$.  Hence, the antenna can move in a 2D $T_1\times T_2$ rectangle.

\textbf{Three-dimensional case:} Here, $\mathbf{t} = (t_1, t_2, t_3)\in A = [0,T_1]\times[0,T_2]\times[0,T_3]$. Hence, the antenna can move in a 3D cuboid of dimensions $T_1\times T_2 \times T_3$.
\vspace{-1.25em}
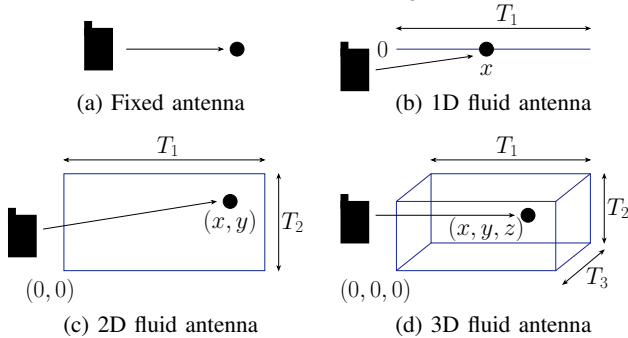
\begin{figure}[!ht]
\centering
\resizebox{0.46\textwidth}{!}{%
\begin{circuitikz}
\tikzstyle{every node}=[font=\Huge]
\node (tikzmaker) [shift={(-1.25, --1.5)}] at (12.25,19.25) {};
\draw [ color={rgb,255:red,17; green,39; blue,146} , line width=1pt ] (3,16.75) rectangle (10.25,13.25);
\draw [ fill={rgb,255:red,0; green,0; blue,0} ] (9,15.75) circle (0.25cm);
\node [font=\Huge] at (9,15) {$(x,y)$};
\node [font=\Huge] at (2.5,12.5) {$(0,0)$};
\draw [<->, >=Stealth] (3,17.25) -- (10.25,17.25);
\draw [<->, >=Stealth] (10.75,16.75) -- (10.75,13.25);
\node [font=\Huge] at (6.75,17.75) {$T_1$};
\node [font=\Huge] at (11.25,15) {$T_2$};
\draw [ fill={rgb,255:red,0; green,0; blue,0} ] (1,15.25) rectangle (2,13.75);
\draw [ fill={rgb,255:red,0; green,0; blue,0} ] (1,15) rectangle (1.25,15.5);
\draw [->, >=Stealth] (2.25,14.75) -- (8.5,15.75);
\draw [ color={rgb,255:red,17; green,39; blue,146} , line width=1pt ] (15,15.75) rectangle (20.75,13.25);
\draw [ color={rgb,255:red,17; green,39; blue,146} , line width=1pt ] (16.25,16.75) rectangle (22,14.25);
\draw [ color={rgb,255:red,32; green,35; blue,142}, line width=1pt, short] (15,15.75) -- (16.25,16.75);
\draw [ color={rgb,255:red,32; green,35; blue,142}, line width=1pt, short] (20.75,15.75) -- (22,16.75);
\draw [ color={rgb,255:red,32; green,35; blue,142}, line width=1pt, short] (15,13.25) -- (16.25,14.25);
\draw [ color={rgb,255:red,32; green,35; blue,142}, line width=1pt, short] (20.75,13.25) -- (22,14.25);
\draw [<->, >=Stealth] (16.25,17.25) -- (22,17.25);
\draw [<->, >=Stealth] (22.5,16.75) -- (22.5,14.25);
\draw [<->, >=Stealth] (22.5,14) -- (21,12.75);
\node [font=\Huge] at (23,15.5) {$T_2$};
\node [font=\Huge] at (19,17.75) {$T_1$};
\node [font=\Huge] at (22.25,13) {$T_3$};
\draw [ fill={rgb,255:red,0; green,0; blue,0} ] (19.75,15.25) circle (0.25cm);
\node [font=\Huge] at (18.25,14.75) {$(x,y,z)$};
\draw [ fill={rgb,255:red,0; green,0; blue,0} ] (13,15.75) rectangle (14,14.25);
\draw [ fill={rgb,255:red,0; green,0; blue,0} ] (13,15.5) rectangle (13.25,16);
\draw [->, >=Stealth] (14.25,15.25) -- (19.25,15.25);
\node [font=\Huge] at (14.25,12.5) {$(0,0,0)$};
\draw [ color={rgb,255:red,32; green,35; blue,142} , line width=1pt ] (15,21.25) rectangle (22,21.25);
\draw [ fill={rgb,255:red,0; green,0; blue,0} ] (18.25,21.25) circle (0.25cm);
\draw [<->, >=Stealth] (15,22) -- (22,22);
\node [font=\Huge] at (19,22.5) {$T_1$};
\node [font=\Huge] at (18.25,20.5) {$x$};
\draw [ fill={rgb,255:red,0; green,0; blue,0} ] (13,21.25) rectangle (14,19.75);
\draw [ fill={rgb,255:red,0; green,0; blue,0} ] (13,21) rectangle (13.25,21.5);
\draw [->, >=Stealth] (14.25,20.5) -- (17.75,21);
\node [font=\Huge] at (14.5,21.25) {$0$};
\node [font=\Huge] at (18.5,19.25) {(b) 1D fluid antenna};
\node [font=\Huge] at (6.5,11.25) {(c) 2D fluid antenna};
\node [font=\Huge] at (18.5,11.25) {(d) 3D fluid antenna};
\draw [ fill={rgb,255:red,0; green,0; blue,0} ] (3.75,22) rectangle (4.75,20.5);
\draw [ fill={rgb,255:red,0; green,0; blue,0} ] (3.75,21.75) rectangle (4,22.25);
\draw [->, >=Stealth] (5.25,21.25) -- (8.75,21.25);
\node [font=\LARGE] at (18.75,20.75) {};
\node [font=\LARGE] at (18.75,20.75) {};
\draw [ fill={rgb,255:red,0; green,0; blue,0} ] (9.25,21.25) circle (0.25cm);
\node [font=\Huge] at (6.5,19.25) {(a) Fixed antenna
};
\node [font=\Huge] at (18.5,18.25) {};
\node [font=\Huge] at (6.5,9) {};
\node [font=\Huge] at (18.75,9) {};
\end{circuitikz}
}%
\vspace{-2.5em}
\caption{Fluid antenna geometries for multiple dimensions.}
\label{fig:dimensions}
\end{figure}

One of the most important physical factors governing the CFAS performance is the correlation of the channel over  $A$. As is typical in the FAS space \cite{kkwong1,kkwong2,FAS_6G,psomas}, we consider isotropic correlation models of the form
\begin{equation}
    \rho(\tau) = \tfrac{1}{2}\mathbb{E}\left[h(\mathbf{t})h^*(\mathbf{t+\Delta})\right],
\end{equation}
where $\tau=||\mathbf{\Delta}||$ is the Euclidean separation between $\mathbf{t}$ and $\mathbf{t+\Delta}$. We further assume that 
$\rho(\tau)\sim 1 -a\tau^2$ as $\tau \rightarrow 0$, which ensures a mean-square differentiable channel (physically reasonable Rayleigh fading). An important parameter related to the correlation is the variance of the channel derivative, which we define as $\lambda_2$. It is known that $\lambda_2=\mathrm{Var}\big[h'(\mathbf{t})\big]=2a$ \cite{jakes}, where in isotropic fading the derivative can be taken in any spatial direction. For the case $\rho(\tau)=J_0(2\pi\tau)$, $\lambda_2=2\pi^2$ \cite{foschini}. Note that all distances (i.e, $\tau$, $T_i$, $i=1,2,3$) are measured in wavelengths throughout the paper.

\section{Analysis}
\label{sec:analysis}
In this section, we look at known results and methods in the fixed antenna and 1D CFAS cases and extend these results into 2D and 3D CFAS scenarios using random field theory.
\vspace{-0.5em}
\subsection{Fixed Antenna}
Here, $X(t)=X(0)$ as $t\in A = \{0\}$ is a fixed point. Hence, $X(t)\sim\mathrm{Exp}({1}/{2})$ and
\begin{equation}
    P_{hs}^{(0)} = P(X(t) \geq u_0) = \mathrm{e}^{-u_0/2}. \label{eq:0D}
\end{equation}
This is a simple Rayleigh fading result at a single point. Note that the superscript 0 is used in this case, as the antenna has no dimensions in which to move.
\vspace{-0.5em}

\subsection{One Dimension}
In 1D, the process is $X(t)$, $t\in[0,T_1]$, and
\begin{equation}\label{eq1d}
   P_{hs}^{(1)} = P\bigg(\sup_{{t}\in [0,T_1]}\left\{X({t})\right\} \geq u_0\bigg).
\end{equation}
For large thresholds, the approach in \cite{davies} can be used to give an asymptotically exact result for \eqref{eq1d} (asymptotic as $u_0\rightarrow\infty$). This approach is detailed in the appendix of \cite{psomas}. Throughout the paper, we denote the asymptotic approximations for the tail probabilities in \eqref{eq:Ppp} in $n$-dimensions as $ \tilde{P}_{hs}^{(n)}$. With this notation, applying the method in  \cite{psomas} to the 1D case gives
\begin{equation}
    \tilde{P}_{hs}^{(1)} = P\left(X(0)\geq u_0\right) + T_1\,\mathrm{LCR}(u_0),
\end{equation}
where $\mathrm{LCR}(u_0)$ is the level crossing rate (LCR) of $X(t)$ across $u_0$. For Rayleigh fading, we have $P\left(X(0)\geq u_0\right)=\mathrm{e}^{-u_0/2}$ as in (\ref{eq:0D}) and, from \cite{jakes}
\begin{equation}
    \mathrm{LCR}(u_0) = \sqrt{\frac{\lambda_2 u_0}{2\pi}}\mathrm{e}^{-u_0/2}.
\end{equation}
Hence, we obtain
\begin{equation}\label{eq:1D}
    \tilde{P}_{hs}^{(1)} = \mathrm{e}^{-u_0/2}\bigg(1+T_1\sqrt{\frac{\lambda_2u_0}{2\pi}}\bigg).
\end{equation}
As we move into the 2D and 3D space, the existing methodology is difficult to use as level crossings are much more complicated in 2D and 3D. Similarly, it is difficult to apply other established techniques from the literature. For example, order statistics are only suitable for discrete FASs while high SNR approximations shift the whole SNR process and thus do not simplify finding the HSP. Also, while extreme value theory can be followed in 1D, it is not directly applicable to higher dimensions. Options were listed in \cite{blake_communications_1986} to extend the LCR approach to higher dimensions, and the expected Euler characteristic (EEC) was suggested as it is defined for random fields in arbitrary dimensions. Additionally, it is simpler than other proposed approaches, such as Slepian processes \cite{blake_communications_1986}. Hence, the EEC method is developed below.
\vspace{-10pt}

\subsection{Arbitrary Numbers of Dimensions}
In random field theory, the established way to approximate $P_{hs}$ is via the EEC, namely, $\tilde{P}_{hs} = \mathrm{EEC}$, when the EEC is computed for a particular field above a threshold, $u_0$. In (15.10.1) of \cite{adler}, we have the result

\begin{equation}
    \mathrm{EEC} = \sum_{j=0}^{\mathrm{dim}(A)} L_j(A)\rho_j(u_0), \label{eq:EEC}
\end{equation}
where $L_j(A)$ are Lipschitz-Killing curvatures of $A$ and $\rho_j(u_0)$ are the Euler Characteristic (EC) densities. The 1D HSP analysis utilizing the LCR provides insights into the EEC. In the 1D example, threshold exceedances occur in isolated intervals of the line where the SNR is above $u_0$. The LCR is based on counting the numbers of these regions. In higher numbers of dimensions, for large $u_0$, exceedances tend to occur in isolated, compact regions (islands in 2D, ``blobs" in 3D). The EEC is a topological approach which essentially provides a normalized count of these regions in an analogous way to the LCR approach.   For a $\chi_k^2$ process, the EEC gives an asymptotically exact result as $u_0\rightarrow\infty$. In our application, we have a $\chi_2^2$ process, $\mathrm{dim}(A)\in\{0,1,2,3\}$ and the sets, $A$, are points, lines, rectangles or cuboids. A challenge in using (\ref{eq:EEC}) is computing the EC densities, but Theorem 15.10.1 in \cite{adler} provides this result for a $\chi_k^2$ process.

\begin{res}
\label{res:ECdensities}
The EC densities for a $\chi_k^2$ random field are
\begin{equation}
    \rho_0(u_0) = P\left(\chi_k^2 \geq u_0\right),
\end{equation}
\begin{multline}
\label{eq:rhoju0}
    \rho_j(u_0) = \frac{u_0^{(k-j)/2}\mathrm{e}^{-u_0/2}}{2\pi^{j/2}\,\Gamma\!\left(\frac{k}{2}\right)2^{(k-2)/2}}\sum_{l=0}^{\lfloor\frac{j-1}{2}\rfloor}\sum_{m=0}^{j-1-2l}1_{\{k\geq j-m-2l\}} \\ \times\!\binom{k\!-\!1}{j\!-\!1\!-\!m\!-\!2l}\frac{(-1)^{j-1+m+l}\,(j-1)!\,u_0^{m+l}}{m!\,l!\,2^l}.
\end{multline}
\end{res}
\noindent Using Result \ref{res:ECdensities}, we now derive the key results of the paper.

\begin{lem}
    \label{lem:Ppp_alldims}
    In 0 dimensions,
    \begin{equation}
        P^{(0)}_{hs} = \mathrm{e}^{-u_0/2}.
    \end{equation}
    In 1 dimension,
    \begin{equation}
        \tilde{P}^{(1)}_{hs} = \mathrm{e}^{-u_0/2}\bigg(\!1+T_1\sqrt{\frac{\lambda_2u_0}{2\pi}}\,\bigg)\!.
    \end{equation}
    In 2 dimensions,
    \begin{equation}
        \tilde{P}^{(2)}_{hs} = \mathrm{e}^{-u_0/2}\bigg(\!\!1\!+\! \sqrt{\frac{\lambda_2u_0}{2\pi}}(T_1\!+\!T_2)\!+\! \frac{\lambda_2(u_0\!-\!1)}{2\pi}T_1T_2\!\!\bigg)\!.\!
    \end{equation}
    In 3 dimensions,
    \begin{align}
        &{\tilde{P}^{(3)}_{hs}} = \mathrm{e}^{-u_0/2}\bigg(\!1\!+\!\sqrt{\frac{\lambda_2u_0}{2\pi}}(T_1\! +\! T_2\! +\! T_3) \!+\! \frac{\lambda_2(u_0\!-\!1)}{2\pi} \notag \\ &\!\times \!(T_1T_2 \!+\! T_1T_3 \!+\! T_2T_3)\! +\!\! \bigg(\!\frac{\lambda_2}{2\pi}\!\bigg)^{\!\!3/2}\!\!\!\!T_1T_2T_3\!\left(\!u_0^{3/2}\!\!-\!3u_0^{1/2}\!\right)\!\!\!\bigg)\!.\!\!
    \end{align}
\end{lem}
\begin{proof}
    As the maximum number of dimensions is 3 and the field is $\chi_2^2$, we substitute $k=2$ into (\ref{eq:rhoju0}) and compute $\rho_j(u_0)$ for $j\in\{0,1,2,3\}$. To compute the Lipschitz-Killing curvatures, we use the result in \cite[p. 324, p. 333]{adler},
    \begin{equation}
        L_j(A) = \lambda_2^{j/2}L_j^E(A), \label{eq:LjA}
    \end{equation}
    where $L_j^E(A)$ are the intrinsic volumes defined in Table \ref{tab:euclideanintrinsicvalues}. Substituting the values of $\rho_j(u_0)$ and using (\ref{eq:LjA}) and Table \ref{tab:euclideanintrinsicvalues} gives the results in Lemma \ref{lem:Ppp_alldims}.
\vspace{-0.5em}
\end{proof}
\vspace{-1em}
\begin{table}[ht!]
\centering
\caption{Euclidean Intrinsic Volumes}
\label{tab:euclideanintrinsicvalues}
\renewcommand{\arraystretch}{1}
\begin{tabular}{|c|c|c|}
\hline
\multicolumn{1}{|c|}{Dimension} & \multicolumn{1}{c|}{Intrinsic Volume}  & \multicolumn{1}{c|}{Physical Meaning} \\ \hline
0D & $L_0^E(A)=1$ & \\ \hline
1D & \begin{tabular}[c]{@{}c@{}}$L_0^E(A)=1$\\ $L_1^E(A)=T_1$\end{tabular} & \begin{tabular}[c]{@{}c@{}} \\ Length \end{tabular} \\ \hline
2D & \begin{tabular}[c]{@{}c@{}}$L_0^E(A)=1$\\ $L_1^E(A)=T_1+T_2$ \\ $L_2^E(A)=T_1T_2$ \end{tabular} &  \begin{tabular}[c]{@{}c@{}} \\ \hspace{-0.4em}$0.5\,\times\,$Boundary Length\hspace{-0.4em} \\ Area \end{tabular} \\ \hline
3D & \begin{tabular}[c]{@{}c@{}}$L_0^E(A)=1$\\ $L_1^E(A)=T_1+T_2+T_3$ \\ \hspace{-0.4em}$L_2^E(A)=T_1T_2 + T_1T_3 + T_2T_3$\hspace{-0.4em} \\ $L_3^E(A)=T_1T_2T_3$ \end{tabular} & \begin{tabular}[c]{@{}c@{}} \\ $2\,\times\,$Caliper Diameter \\ $0.5\,\times\,$Surface Area \\ Volume \end{tabular} \\ \hline 
\end{tabular}
\end{table}
\vspace{-0.5em}
Note that Lemma \ref{lem:Ppp_alldims} agrees with the separate characterization in \eqref{eq:0D} and \eqref{eq:1D} for the fixed antenna and 1D cases and extends the approach to the 2D and 3D  cases. Using the results in Lemma \ref{lem:Ppp_alldims}, we can state the approximate scaling laws below.

\begin{lem}
    \label{lem:dimscaling}
    The $n$-th dimension scales the value of $\tilde{P}_{hs}$ by approximately $1+T_n\sqrt{\frac{\lambda_2u_0}{2\pi}}$.
\end{lem}
\begin{proof}
    From Lemma \ref{lem:Ppp_alldims}, we observe that moving from a fixed antenna to a linear fluid antenna gives
    \begin{equation}
        \tilde{P}^{(1)}_{hs}=P^{(0)}_{hs}\times\left(1+T_1\sqrt{\frac{\lambda_2u_0}{2\pi}}\right).
    \end{equation}
    Hence, the scaling is exact in 1D. Moving to 2D, some simple algebra gives
    \begin{equation}
        \tilde{P}_{hs}^{(2)} = \tilde{P}_{hs}^{(1)}\times\left(1+T_2\sqrt{\frac{\lambda_2u_0} {2\pi}}\right)+ R_2,
    \end{equation}
    where the remainder term is $R_2 = -T_1T_2\pi\mathrm{e}^{-u_0/2}$. Note that for large $u_0$, $R_2$ is considerably smaller than the leading term. Similarly, for 3D,
    \begin{equation}
        \tilde{P}_{hs}^{(3)} = \tilde{P}_{hs}^{(2)}\times\left(1+T_3\sqrt{\frac{\lambda_2u_0}{2\pi}}\right) + R_3,
    \end{equation}
    where $R_3 = -T_3\pi\mathrm{e}^{-u_0/2}\left(T_1+T_2+2T_1T_2\sqrt{\pi u_0}\right)$. Again, $R_3$ is small compared to the leading term for large $u_0$.
\end{proof}

The accuracy of  Lemmas \ref{lem:Ppp_alldims} and \ref{lem:dimscaling} is demonstrated in Section \ref{sec:numresults}. The scaling laws offer remarkably simple insights into CFAS performance. The $n$-th dimension scales $\tilde{P}_{hs}$ by a scalar multiplicative factor which is linear in $T_n$, proportional to $\sqrt{u_0}$ and inherently incorporates spatial correlation. Note that this law is valid when the scaled result remains a small probability to ensure compliance with probability bounds.

In a symmetric scenario, where $T_1 = T_2 = T_3 = T$, we have $\tilde{P}_{hs}^{(n)}\approx P_{hs}^{(0)}\Big(1+T\sqrt{\tfrac{\lambda_2 u_0}{2\pi}}\Big)^n$. As an example, consider the Bessel function correlation with $\lambda_2=2\pi^2$, and the parameters $T=2$ (2 wavelengths), and $u_0=6.4$. For this scenario, 

\vspace{-0.4em} 
\begin{equation}
    \tilde{P}_{hs}^{(n)}\approx P_{hs}^{(0)}\times 10^n.
\end{equation}
Hence, each dimension gives a ten-fold improvement in the probability of exceeding $u_0=6.4$.

\subsection{Optimal Shapes}
A key question concerning 2D and 3D CFASs is the nature of the optimal shape. We define optimality as maximizing $\tilde{P}_{hs}^{(n)}$. In 1D, we simply take the longest line available, i.e. the largest $T_1$ available, as $\tilde{P}_{hs}^{(1)}$ is linearly increasing in $T_1$. In 2D, we could have a size constraint on the surface area ($S$), such that $T_1T_2\leq S$ and individual limits on the two dimensions, such that $T_1 \leq L_1$ and $T_2 \leq L_2$ where $L_1L_2 > S$\footnote{Note that a single area constraint is insufficient as the corresponding maximum is an infinitely long, infinitely thin surface.}. Hence, we are able to select a sub-rectangle of area $\leq S$ within $[0,L_1]\times[0,L_2]$ to enhance the performance. Similarly, in 3D, consider the volume constraint $T_1T_2T_3\leq V$ and individual limits $0\leq T_i\leq L_i$, $i=1,2,3$. Here, we can select a sub-cuboid of volume $\leq V$ within $[0,L_1]\times[0,L_2]\times[0,L_3]$. The optimal shapes are given in the following Lemma.

\begin{lem}
    \label{lem:optshapes}
    Without loss of generality, we assume that $L_1\leq L_2\leq L_3$. In the 2D case, the optimal sub-rectangle of dimension $T_1\times T_2$, which maximizes $\tilde{P}_{hs}^{(2)}$, is given by $T_2=L_2$, $T_1={S}/{L_2}$ for the following conditions: $0\leq T_i \leq L_i$, $i=1,2$, $T_1T_2\leq S$. In the 3D case, the optimal sub-cuboid of dimensions $T_1\times T_2\times T_3$ which maximizes $\tilde{P}_{hs}^{(3)}$ is given by $T_3=L_3$, $T_2=L_2$, $T_1={V}/{(L_2L_3)}$ for the following conditions: $0\leq T_i\leq L_i$, $i=1,2,3$, $T_1T_2T_3\leq V$.
\end{lem}
\begin{proof}
    We give the proof for the 2D case. The 3D case follows similar steps, but is lengthy and is outlined in the proof sketch in Appendix A. From Lemma \ref{lem:Ppp_alldims}, we have $\tilde{P}_{hs}^{(2)}=a+b(T_1+T_2)+cT_1T_2$ for positive constants $a,
    b,c$. As $\tilde{P}_{hs}^{(2)}$ is monotonically increasing in $T_1$ and $T_2$, it follows that the area constraint $T_1T_2\leq S$ becomes an equality, $T_1T_2=S$. Hence, $T_1 = {S}/{T_2}$ and
    \begin{equation}
        \tilde{P}_{hs}^{(2)} = a+b\left(\tfrac{S}{T_2}+T_2\right)+cS.
    \end{equation}
    Hence, optimizing $\tilde{P}_{hs}^{(2)}$ is equivalent to maximizing $g(T_2)=\tfrac{S}{T_2}+T_2$ over the limits $0\leq \tfrac{S}{T_2}\leq L_1$, $0\leq {T_2}\leq L_2$. Note that $g(T_2)$ has a single turning point at $T_2=\sqrt{S}$ for $T_2\geq0$ and this is a minimum as $g''(T_2)=\tfrac{2S}{T_2^3}>0$. Hence, the maximum value of $g(T_2)$ occurs at an edge point, $T_2\in\{\tfrac{S}{L_1},L_2\}$. Evaluating $g(\cdot)$ at both points shows that $g(L_2)>g(\tfrac{S}{L_1})$. Hence, the optimal shape is $T_2=L_2$, $T_1=\frac{S}{L_2}$.
\end{proof}

We note that the optimal sub-rectangle and the optimal sub-cuboid are the \textit{least} compact shapes possible. The dimensions are as far away from a square in 2D and a cube in 3D as possible. Hence, in 2D, the dimension with the highest limit is maximized ($T_2=L_2$) and then the second dimension is fixed by the surface area constraint ($T_1=\tfrac{S}{T_2})$. In 3D, the two dimensions with the highest limits are maximized ($T_3=L_3, T_2=L_2$) and the third dimension is set by the volume constraint ($T_1=\tfrac{V}{T_2T_3}$). While this result is intuitive, it is not always simple to prove generically, and the methods of analysis used in this work have allowed us to do so. 

\section{Numerical Results}\label{sec:numresults}
This section verifies the analytical results and explores the
system behavior. As discussed in Section \ref{sec:sysmodel}, the correlation over distance $\tau$ is represented by the classical Jakes' model, where $\rho(\tau)=J_0(2\pi\tau)$ \cite{foschini}. Although this model is used in the simulations, the analysis covers all correlation models of the form $\rho(\tau)\sim 1 -a\tau^2$ as $\tau \rightarrow 0$ . For simulated results, channels are generated every $0.01\lambda$ across each dimension over the possible antenna space to identify the location with the maximum SNR, and $10^6$ replicates are generated.
\vspace{-0.5em}

\subsection{High SNR Probability and Dimensional Scaling}

Figure \ref{fig:quarterwavelength} validates the analytical HSP results in Lemma \ref{lem:Ppp_alldims}. The logarithmically scaled complementary CDF is plotted for a  CFAS with $n=\{0, 1, 2, 3\}$ dimensions, where $T=T_1=T_2=T_3=0.25\lambda$. These results are compared to the dimensional scaling law proposed in Lemma \ref{lem:dimscaling}, $\tilde{P}_{hs}^{(n)} \approx P_{hs}^{(0)}(1+T\sqrt{\tfrac{\lambda_2 u_0}{2\pi}})^n$, highlighting its accuracy. 

The small value of $T=0.25\lambda$ was selected for these results due to the computational challenges of the 3D scenario. Channels must be generated for a sufficient number of antenna locations within the region of interest, $A$, to accurately simulate the HSP of a CFAS. All of these channels are correlated, which requires the creation of very large correlation matrices. For example, for a 3D CFAS where channels are generated at 27 equally spaced points in each dimension, the correlation matrix has dimensions of $27^3\times 27^3 = 19,683 \times 19,683 $, a matrix with a total of $7.6\times 10^{12}$ elements. If each element requires 8 bytes of storage, this matrix fills nearly 16 GB of RAM. The square root of this matrix is also required, as well as $3.8\times 10^{13}$ matrix multiplications to compute all $10^6$ replicates. This leads to both time and storage issues with standard software. The challenges above particularly highlight the benefits of analytical results for 3D CFASs.

\begin{figure}[ht]
    \centering
    \includegraphics[scale=0.55,trim={0.93cm 0.05cm 1.5cm 0.67cm},clip]{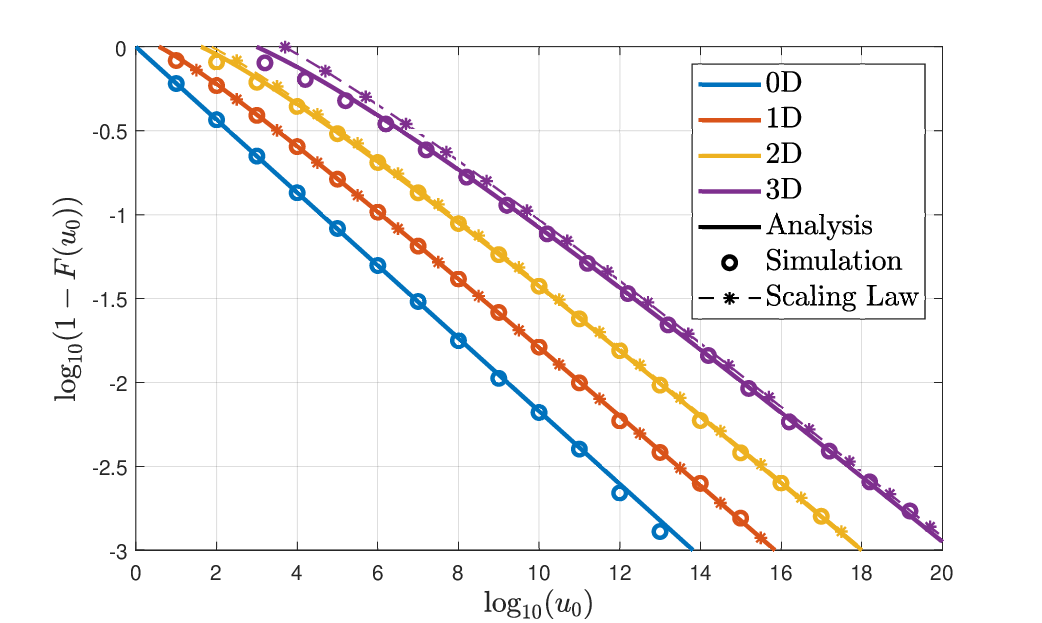}
    \caption{A comparison of the analytical, simulated and scaling law HSPs for 0-3 dimension CFAs, where each dimension is set to $0.25\lambda$.}
    \label{fig:quarterwavelength}
\end{figure}

Figure \ref{fig:quarterwavelength} shows that the analytical and simulated values match well, particularly in the small values of the upper tail (such as below $0.1$). Additionally, the scaling laws closely match the full $\tilde{P}_{hs}^{(n)}$ results, proving their accuracy in assessing the impact of increasing CFAS dimensions.

Finally, the enhancement offered from increasing the dimension of the CFAS is clear. This is intuitive from a physical perspective - more dimensions lead to more available antenna positions, which increases the likelihood of the existence of a position resulting in performance above an SNR threshold. 
\vspace{-0.5em}

\subsection{Impact of Antenna Shape}
Figure \ref{fig:dims} investigates the impact of varying the antenna shapes in 2D and 3D scenarios. The surface area of the 2D CFAS and the volume of the 3D CFAS are fixed to $1$, and a range of side lengths are considered. 

\begin{figure}[ht]
    \centering
    \includegraphics[scale=0.55,trim={0.93cm 0.05cm 1.5cm 0.67cm},clip]{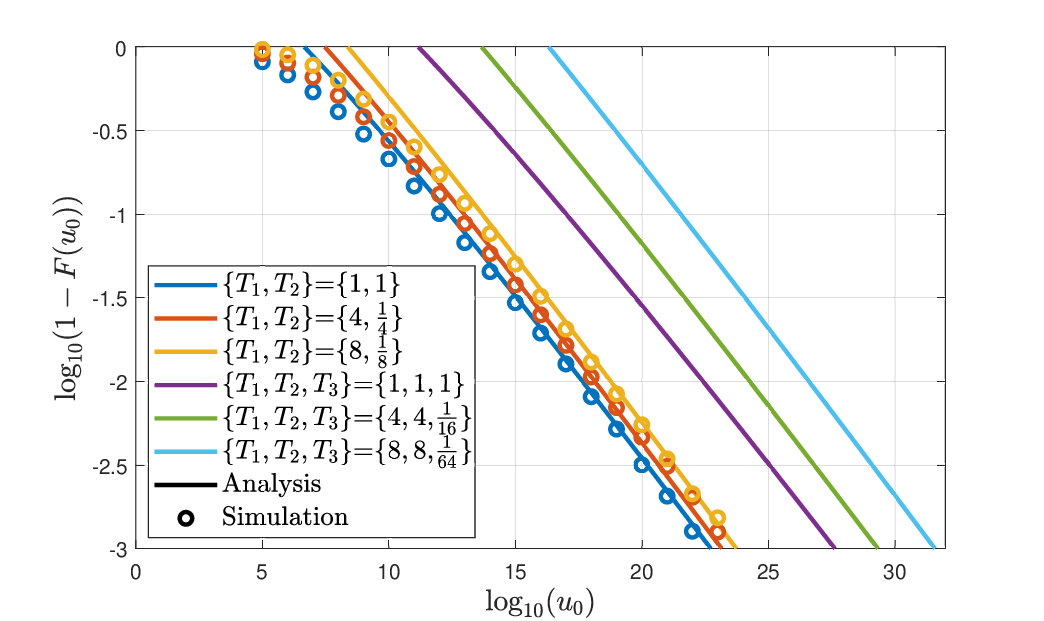}
    \caption{The analytical HSP for 2 and 3 dimensional CFASs, where each dimension is set in multiples of $\lambda$ according to the legend.}
    \label{fig:dims}
\end{figure}

As predicted in Lemma \ref{lem:optshapes}, the HSP is higher for ``non-compact" shapes of $A$ when compared to compact, symmetric structures. For example, in the 2D scenario, the HSP is higher when $T_1=8$ and $T_2=\frac{1}{8}$ ($A$ is a very long, thin rectangle) than when $T_1=T_2=1$ ($A$ is a square). Similarly, in the 3D scenario, the HSP is higher when $A$ is made very large in two dimensions than when $A$ is a cube. Non-compact structures increase the Euclidean distance between opposite vertices of $A$. This leads to greater spatial diversity and thus increases the probability that a point in $A$ results in an SNR $>u_0$.

Simulated results are shown for the 2D scenario. Again, the simulations are very accurate for HSP values below $0.1$. Compared to Fig.~\ref{fig:quarterwavelength}, we consider a larger volume for $A$  in the 3D scenario. Due to the computational challenge of simulating enough channels to represent the larger space accurately, simulations are not produced here. The reliance on analytical expressions to explore the 3D scenario highlights their importance again. Furthermore, the expected improvement in HSP occurs as the dimensionality of $A$ increases. While the total geometric measure remains fixed, the third dimension introduces a new degree of freedom and thus an increase in the possible antenna positions.
\vspace{-0.5em}

\section{Conclusion}
In this work, we derived novel, asymptotically exact expressions for the HSP of spatially coherent isotropically correlated 2D and 3D CFASs using the EEC from random field theory. Using known results for the HSP of fixed antennas and 1D CFASs and the derived 2D and 3D HSP expressions, we proposed a remarkably accurate approximate dimensional scaling law that relates the HSP of a fixed antenna to that of an $n$-dimensional CFAS, providing useful insights into fundamental system behavior. Finally, we showed that the optimal shape for 2D rectangular and 3D cuboidal CFASs is the least compact shape allowed by individual dimensional limits. This allows for more spatial diversity in the system.

\section*{Appendix A \\ Sketch Proof for the 3D scenario in Lemma \ref{lem:optshapes}}
As the objective function is increasing in $T_1, T_2$ and $T_3$, the constraint $T_1 T_2 T_3 \le V$ becomes the equality $T_1 T_2 T_3= V$. Next, for positive constants $a, b, c, d$, from Lemma \ref{lem:Ppp_alldims}, we have
    \begin{equation}\notag
        \tilde{P}_{hs}^{(3)}=a\!+\!b\left(\tfrac{V}{T_2T_3}\!+\!T_2\!+\!T_3\right)\!+\!c\left(\tfrac{V}{T_3}\!+\!T_2T_3\!+\!\tfrac{V}{T_2}\right)\!+\!dV.
    \end{equation}
    As in the 2D case, stationary points inside the allowable region $\{0 < T_2 \le L_2,0<T_3<L_3\}\cup\{T_2T_3 \ge \frac{V}{L_1}\}$ yield minima and the maxima are attained at the edges. Close inspection of the edges shows that $T_2=L_2$ and $T_3=L_3$ are optimal.
\bibliographystyle{IEEEtran}

\end{document}